\documentclass[letterpaper]{article}

\usepackage{graphicx}
\usepackage{hyperref}

\usepackage[margin=1.2in]{geometry}
\usepackage{lmodern}   % fixedwidth font
\usepackage[T1]{fontenc} % some website told me to do this.

\usepackage{amsthm} % Needed for \boldsymbol
\usepackage{amsmath} % Needed for \boldsymbol

\usepackage{natbib} 
\newcommand{\Vnti}{\mathbf{V}}
\newcommand{\Del}{\boldsymbol{\Delta}}
\newcommand{\DelT}{\Del'}
\newcommand{\J}{(\DelT \Vnti \Del)}

\DeclareMathOperator{\vech}{vech}

\newcommand{\IPC}{\mathbf{t}}
\newcommand{\var}{\mathrm{var}}
\newcommand{\EPC}{\textrm{EPC}}

\newcommand{\pvec}{\boldsymbol{\theta}}

\newcommand{\y}{\mathbf{y}}
\newcommand{\z}{\mathbf{z}}
\renewcommand{\d}{\mathbf{d}}

\newcommand{\W}{\mathbf{W}}

\newcommand{\proglang}{\texttt}

\title{Individual Differences in Structural Equation Model Parameters}
\author{DL Oberski}
\date{}

\begin{document}
\maketitle

\begin{abstract}
Individuals may differ in their parameter values.  This article discusses a three-step method of studying such differences by calculating and then modeling ``individual parameter contributions'', making the study of heterogeneity in arbitrary structural equation model parameters as technically challenging as performing linear regression. The proposed approach allows for a contribution to the study of differential error variances in survey methodology that would have been difficult to make otherwise.
\end{abstract}

\section{Introduction}

Structural equation models assume that the same parameter values apply to all individuals. This assumption is probably false. Such ``aggregated''  analyses \citep{muthen_complex_1995} may still provide useful models  in some cases, particularly when individual differences in parameter values are not of substantive interest. In other cases,  interest focuses precisely on individual differences in, for instance, latent variable means, regressions of one set of latent variables on another, or latent or measurement error variances. Of these, it is trivially simple to compare latent variable means across values of another variable: a  regression of the latent variable on an independent variable will do. Most structural equation models involve at least one such regression \citep[e.g.][]{maccallum2000applications}. Multiple independent variables' effects on the latent mean can be explored just as simply. Less simple, however,  is the comparison of other types of parameters such as latent regression coefficients, variances, and measurement error variances.  Such comparisons are, however, often also of scientific interest.

%For example, \citet{davidov2008values} regress citizens' latent attitudes towards immigration on their reported basic values across 19 European countries. Estimating a 19-group structural equation model with 323 cross-group equality restrictions, they estimate different regression coefficients of interest in each country, and go on to correlate the obtained country regression coefficient estimates with the country-level covariates ``gross domestic product'' and ``inflow of immigrants''. Not only countries, but also regions and indeed people have differing incomes, so that beyond this already advanced analysis there is reason to think that an exploration of differences in the latent variable regression coefficients with respect to other covariates may also have been of interest. This would have complicated the analysis non-trivially, however.

%In a different
Examples of applications that focus on differences in structural parameters include \citet{davidov2008values}'s comparison of the regression coefficients of latent attitudes to immigration on latent values across 19 European countries, and \citet{kendler2000illicit}'s study of heritability of marijuana use in twins, which compared variance estimates of additive genetic factors between ``use'' and ``heavy use'' groups. 
Individual differences in reliability and validity of survey questions are another example. \citet{revilla2012impact}  split the sample into groups on one  covariate at a time and estimated multitrait-multimethod models in each group. \citet{scherpenzeel_validity_1997} used the same approach, as  did \citet[pp. 168--172]{alwin_reliability_1991} and \citet{alwin_margins_2007}.
Instead of looking at the impact of  age, gender, and education controlling for the other factors, only the marginal impact of each of these factors could be assessed \citep[p. 52]{revilla2012impact}. While it is in principle possible to split the sample by the interaction of these three covariates, this would have resulted in groups far too small for estimating the multitrait-multimethod parameters of interest. Moreover, to allow for multiple-group structural equation modeling, the ``age'' variable required  categorization. 

Solutions to these complications exist. However, relating possibly latent variances, covariances, and regression coefficients to explanatory variables is  not currently as simple an exercise as regressing an outcome on a set of predictors. 
As demonstrated in the examples, multiple group structural equation modeling \citep{sorbom1974general,bollen_structural_1989,raykov2006first}, though flexible and entirely general, can lead to a large number of groups, parameters, and equality restrictions. Interaction modeling is a possible alternative when differences in regression coefficients are of interest \citep{kenny1984estimating,bollen1995structural,jaccard1996lisrel,joreskog1996nonlinear,schumacker1998interaction,preacher2006computational,van2013testing}. Interaction modeling has the drawback that it is not well-suited to studying differences in variances; and  with more than one predictor, the dimensions of integration \citep{klein2000maximum} or selection of appropriate product indicators or  third-order moments \citep{mooijaart2010alternative} is an additional complication.
Another possible avenue is multilevel structural equation modeling \citep{muthen1989latent,muthen1994multilevel,muthen_complex_1995} with random slopes, which has become more convenient with recent software advances \citep{muthen2012bayesian}. This method is flexible but again focuses on slopes rather than variances and has the drawback that the random effects models require a normality assumption. Random slopes models, moreover, require numerical integration,  precluding problems with many varying coefficients and a large number of observations.  
Currently available methods for looking at differences in structural equation model parameters are therefore well-developed, but not highly convenient for exploration with several candidate covariates, parameter differences, and big data sets.

This article proposes a three-step procedure that considerably simplifies relating structural equation model parameters to covariates. In step one, a single-group structural equation model is fitted. In step two, based on the step one estimates, the observed data are transformed into a new data set of ``individual parameter contributions'' (IPC's). These individual parameter contributions can be read in by any basic statistics program and regressed on any set of covariates in step three. 

The proposed three-step IPC regression procedure is closely related to the estimation of modification indices and expected parameter changes \citep{saris_detection_1987}.
We show that a chi-square test for differences in the individual parameter contributions with respect to a grouping variable is equivalent to the modification index that would have been obtained in a hypothetical multiple group model with equality restrictions. However, the researcher requires only the ``aggregated'' analysis to obtain the individual parameter contributions. Exploring individual differences in arbitrary structural equation model parameters thus becomes as simple as running a multiple linear regression or $t$-test on the individual parameter contributions.

Testing of IPC regression coefficients is also closely related to the econometric literature on ``structural change''
\citep{kuan1995generalized,hjort2002tests}. ``Structural change'' tests  were originally motivated by a search for changes over time in times series model parameters \citep{brown1975techniques,nyblom1989testing,andrews1993tests}. 
\citet{zeileis2005unified} and \citet{zeileis2007generalized} generalized this econometric technique to maximum likelihood estimation and \citet{merkle2013tests} discussed its application to the investigation of measurement invariance. The test obtained from IPC regression regression is proportional to a subset of the structural change tests that \citet[p. 123]{hjort2002tests} show to have optimal asymptotic power. Whereas the structural change approach places emphasis on hypothesis tests, however, IPC regression also estimates the differences' magnitude. In other words, when between-group differences are investigated, IPC regression provides not only the modification index (Lagrange multiplier, score test), but also the expected parameter change. Moreover, 
IPC regression allows for modeling of the parameter differences in the third step: for instance, differences with respect to several variables at a time may be evaluated, as the application shows.

A field of application where IPC regression can be useful is survey methodology: survey researchers such as  \citet{alwin_reliability_1991}, \citet{scherpenzeel1995question}, \citet{alwin_margins_2007}, and \citet{revilla2012impact} have studied how survey question reliability varies with respondent characteristics. Past research estimated reliability, or measurement error variance, with a multiple group structural equation model and compared the result across groups, a method that does not permit studying the effects of continuous respondent characteristics or the simultaneous evaluation of more than one respondent characteristic at a time. We contribute to this field by showing how IPC regression can be applied to regress measurement error variance in a ``quasi-simplex'' model of survey error on respondent characteristics. Effects of Age, Gender, and Education level on the measurement error variance in answers to the frequency question ``how often do you use the internet?'' are found, but these effects dissappear when controlling for a variable indicating that the respondent never uses the internet.  We discuss some possible implications for the analysis of frequency questions in surveys.

\vspace{12pt}The remainder of this article is structured as follows.
The data transformation yielding individual parameter contributions from an estimated single-group model is derived in the following section. As is the case for expected parameter changes and modification indices, regressing individual parameter contributions on covariates yields only approximately unbiased estimates. The subsequent  simulation study therefore not only evaluates the performance of the suggested method, but also compares performance to that of the commonly applied multiple group structural equation model. The application is then discussed: after fitting a four-wave quasi-simplex model to a longitudinal probability sample of respondents, we assess individual differences in the measurement error variance of answers to a survey question. The final section discusses the results and concludes.

\section{Regressing SEM estimates on covariates}\label{sec:theory}

Intuitively, structural equation model parameters are estimated by combining the observations. Although the way in which each observation contributes to a parameter estimate is not usually linear, it is possible to obtain a linear approximation \citep{bentler1984efficient} to each observation's contribution. Sampling fluctuations will mean the contributions differ over observations, but differences between parameter contributions may also be due to true (``structural'') differences in the parameters between observations. This section  formally motivates individual parameter contributions and  shows that performing linear regression of the individual parameter contributions on a covariate is equivalent to regressing the SEM parameter directly on the covariate, if it were possible to formulate such a model. The usual regression coefficient hypothesis tests are valid tests of the relationship between a covariate and a SEM parameter. Since the independent variable in the linear regression may be a set of dummy variables, the results also apply to comparisons between groups. In that case, as the appendix shows, IPC regression is equivalent to calculating versions \citep{bentler1992some} of the more familiar Expected Parameter Change and Modification Index statistics for differences between groups in SEM parameters.

Given a vector $\y$  of $p$ observed variables with population covariance matrix $\boldsymbol{\Sigma}$, a
structural equation model can be viewed as a covariance structure model 
$\boldsymbol{\Sigma} = \boldsymbol{\Sigma}(\pvec)$, where $\boldsymbol{\Sigma}(\pvec)$ is a continuously
differentiable matrix-valued (symmetric and positive definite) function of the
vector of parameters $\pvec$ of the model  
\citep[see, e.g.][for more details]{bollen_structural_1989}.

Given an observed covariance matrix $\mathbf{S}$, based on a sample of size $n$ of $\y$,
the vector of parameters $\pvec$ is estimated by minimizing with respect to
$\pvec$ a discrepancy function $F(\pvec) = F(\mathbf{S}, \boldsymbol{\Sigma}(\pvec))$ of $\mathbf{S}$ and
$\boldsymbol{\Sigma}(\pvec)$. 
This approach encompasses maximum likelihood as well as weighted least squares
estimators.

A key matrix is the hessian matrix $\mathbf{V}
= \frac{1}{2} \frac{\partial^2 F(\mathbf{S}, \boldsymbol{\Sigma})}{\partial \boldsymbol{\sigma} \partial \boldsymbol{\sigma}'}$, where
$\boldsymbol{\sigma} := \vech(\boldsymbol{\Sigma})$ is the half vectorization of $\boldsymbol{\Sigma}$. In weighted least squares estimators, 
$\mathbf{V}$ is a weight matrix determined by the choice of estimator
\citep{satorra1989alternative}. To study how parameter values vary as a function of observations, we also define $n$ vectors of size $p (p + 1)/2$ 
\begin{equation}
	\d_i := \vech\left[ (\y_i - \bar{\y}) (\y_i - \bar{\y})'\right],
	\label{eq:di}
\end{equation}
\citep[e.g.][]{satorra_asymptotic_1992}, so that $\bar{\d} = \mathbf{s}$. 
In the case of complex sample data, $\bar{\y}$ should be replaced with a consistent (weighted) estimate of the population mean \citep[p. 283]{muthen_complex_1995}.

Of interest are differences in the parameter values across values of a vector of covariates $\z$: it may be that the parameter vector $\pvec_i$ varies per observation $i$ due to differences in $\z_i$. Hypothetically, if each $\pvec_i$ could be directly observed, studying differences would simply amount to a regression such as the linear regression $\pvec_i = \boldsymbol{\gamma} \z_i$. For instance, if $\z$ is a dummy-coded grouping variable, differences across groups in the structural equation model parameters would be reflected in the regression coefficient vector $\boldsymbol{\gamma}$. If observations $\hat{\pvec}_i$ containing sampling error were obtained, sample estimates of the regression model could be obtained by performing, for instance, a linear regression 
\begin{equation}
	\hat{\pvec}_i = \boldsymbol{\gamma} \z_i + \boldsymbol{\delta}_i.
	\label{eq:regression}
\end{equation}
This would yield sample estimates 
$\hat{\boldsymbol{\gamma}} = (\z\z')^{-1}\z \hat{\boldsymbol{\Theta}}$, where $\boldsymbol{\hat{\Theta}} := (\hat{\pvec}_1, ..., \hat{\pvec}_n)'$.

In practice 
$\boldsymbol{\hat{\Theta}}$ is not observed, but it is still possible to approximate 
$\hat{\boldsymbol{\gamma}}$ by creating a transformed data matrix and then performing linear regression on $\z$. 
From the above, it is clear that 
$$
\hat{\pvec} = f(\d, \z).
$$
And from Equation \ref{eq:regression} and by applying the implicit function theorem, 
$$
	\boldsymbol{\gamma} = \frac{\partial \hat{\pvec}} {\partial \z} = 
		- \left(\frac{\partial^2 F(\mathbf{S}, \boldsymbol{\Sigma}(\pvec)}{\partial \pvec\partial\pvec'}\right)^{-1} \left(\frac{\partial F(\mathbf{S}, \boldsymbol{\Sigma}(\pvec)}{\partial \pvec \partial \z'}\right).
$$
From \citet{satorra1989alternative}, it can be seen that the asymptotic limits of 
$\frac{\partial^2 F(\mathbf{S}, \boldsymbol{\Sigma}(\pvec)}{\partial \pvec\partial\pvec'}$
and $\frac{\partial F(\mathbf{S}, \boldsymbol{\Sigma}(\pvec)}{\partial \pvec \partial \z'}$ are
$\Del' \mathbf{V} \Del$ and $\Del' \mathbf{V} (\partial \d / \partial \z')$ respectively,
letting $\Del := \partial \boldsymbol{\sigma} / \partial \pvec$, which can be obtained
from the parameter values using the expressions given by \citet{neudecker1991linear}. 
Given model (\ref{eq:regression}), $\partial \d / \partial \z'$ is the coefficient  in a linear regression. 

The effects of $\z$ on parameter values $\pvec$ can therefore be obtained as 
\begin{equation}
\boldsymbol{\hat{\gamma}} =	(\z \z')^{-1} \z \cdot \W\d ,
	\label{eq:solution}
\end{equation}
where the transformation matrix $\W := -\J^{-1} \Del' \mathbf{V}$.
In a given sample, an approximately consistent estimate $\hat{\W}$ of $\W$ can be obtained 
by replacing the model parameters that determine $\Del$ and $\mathbf{V}$ by their sample estimates.
In other words, after transforming the observed data as
\begin{equation}
	\hat{\IPC} := \hat{\W} \cdot \d,   \label{eq:IPC-def}
\end{equation}
the $\boldsymbol{\gamma}$ coefficient vector can be estimated simply by
regressing $\hat{\IPC}$ on $\z$, or by any other linear model for $\hat{\IPC}$ as
a function of $\z$. Since the $\hat{\IPC}$ values represent each observation's 
contribution to the linearized sample parameter estimate, we will refer to 
$\hat{\IPC}$ as ``individual parameter estimate contributions'' (IPC's).

The variance estimate of $\hat{\boldsymbol{\gamma}}$ obtained in this way
will be consistent as well, and ``Wald'' or $t$-tests will be valid. Moreover,
standard errors and test statistics obtained from the regression of $\hat{\IPC}$ on $\z$ will be robust to departures of $\y$ from normality. This
can easily be seen  by recognizing $\hat{\var}(\hat{\pvec}) = n^{-1} (\hat{\IPC} - \bar{\hat{\IPC}})'(\hat{\IPC} - \bar{\hat{\IPC}})$ as the robust parameter estimate
covariance matrix \citep{fuller1987, satorra1994corrections}. 
Therefore, the method proposed here is distribution-free.
Complex sampling designs can be taken into account by taking account of clustering and stratification in the variance calculations for the linear regression of $\hat{\IPC}$ on $\z$, something all standard statistical packages allow the user to do. 

In the special case where $\z$ is a  dummy variable for a grouping of interest,
a sample ``Wald'' test of $H_0: \hat{\gamma} = 0$ is shown in Appendix \ref{sec:EPC} to be equivalent to calculating a ``generalized'' modification index in a multiple-group structural equation model  \citep[section 5]{saris_detection_1987,satorra1989alternative}. Consider the situation where a multiple-group structural equation model is formulated, using $\z$ as the grouping variable, 
with cross-group equality constraints on all parameters including the parameter of interest. A ``modification index'' is then calculated for the statistical significance of the change in the parameter of interest when freeing all parameter equality constraints \citep{bentler1992some}. This modification index will equal the ``Wald'' test for $H_0: \hat{\gamma} = 0$ in the simple regression of $\hat{\IPC}$ on $\z$. In addition,   
$\hat{\gamma}$ will equal the corresponding ``expected parameter change'' (depending on the parameterization).

Just as for modification indices and expected parameter change coefficients, in deriving
approximate consistency of $\hat{\boldsymbol{\gamma}}$ the assumption is made that the second derivative matrix
$\mathbf{J} := \Del' \mathbf{V} \Del$ is approximately constant between the null and the true model \citep{saris_detection_1987}. 
In the case discussed here, the relevant assumption is that $E_\z [ \mathbf{J}(\hat{\pvec} | \z)] \approx \mathbf{J}$. This assumption will not hold in general, since
the average conditional parameter estimate is not generally equal to the parameter estimate
for a completely pooled model \citep{muthen1989latent}. However, it will hold approximately for parameters depending on the covariance matrix such as latent regression coefficients and variances when the sample covariance matrix is conditioned on $\z$, i.e. when $\z$ is included in the model as a fixed covariate with direct effects on all the observed variables.

\vspace{12pt}
In summary, possible differences in structural equation model parameter estimates  
can be investigated by modeling a transformed data set $\hat{\IPC}$ as a linear function of $\z$, for instance by OLS linear regression or a $t$-test
for differences between groups. 
This suggests a three-step approach to exploring differences in any set of parameters of a structural equation model with respect to a set of covariates:
\begin{enumerate}
	\item Estimate the overall single-group model (preferably introducing $\z$ as fixed covariates);
	\item From the data and step 1, obtain the transformed data set $\hat{\IPC}$;
	\item Regress the individual parameter contributions $\hat{\IPC}$ for the parameter(s) of interest on the covariates.
\end{enumerate}

Appendix \ref{sec:Rcode} provides an \proglang{R} function (\texttt{get.ipc}) that generates the transformed IPC data set ($\hat{\IPC}$), given a data set and a structural equation model fitted with the \texttt{lavaan} package \citep{Rlanguage,rosseel2013lavaan}. The transformed data can then be written to a file and analyzed by basic statistics software or analyzed directly from within \proglang{R}.
An example analysis using these functions is given in Appendix \ref{sec:Rcode-example}.

\section{Simulation study\label{sec:simulation}}

To evaluate the finite-sample performance of the IPC regression proposed here, we set up a small simulation study. The study has two goals: 
\begin{enumerate}
\item To evaluate the overall performance of the method in terms of bias, coverage, type I error, and power in a simple two-group setting; 
\item To compare performance of the method of taking IPC mean differences with that of the currently most applied, and asymptotically most efficient, method: multiple group structural equation modeling.
\end{enumerate}

\begin{figure}
\centering
	\includegraphics[width=.55\textwidth]{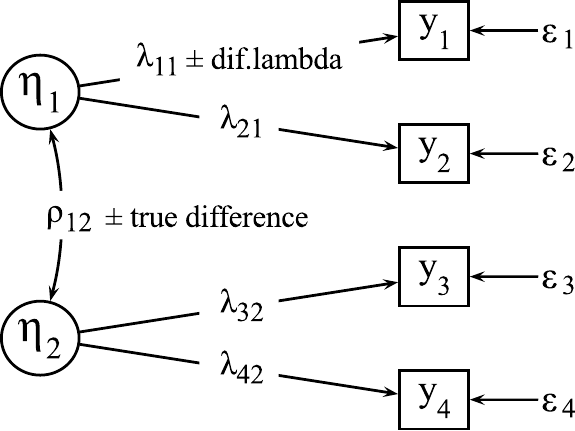}

	\caption{Simulation model. Of interest is the true between-group difference in the factor correlation, $\rho_{12}$. Depending on the condition, there is also a difference of  ``dif.lambda'' in the size of the first loading. Between-group differences in the intercepts of $y_1$ -- $y_4$ are always present (not shown for clarity).}\label{fig:simulation-model}
\end{figure}

Figure \ref{fig:simulation-model} shows the two-factor model (1 degree of freedom) used for our Monte Carlo simulations.  A two-group population is generated, with all loadings equal to $1$ in group 1, and $\lambda_{21} = \lambda_{42} = 1, \lambda_{32} = 1/2$, and $\lambda_{11} = 1 + \textit{dif.lambda}$ in group 2. All error variances in both groups equal $\var(\epsilon_{j}) = 0.8$. The intercepts of all observed variables equal $1$ in group 1 and $2$ in group 2. The latent variables $\eta_1$ and $\eta_2$ are standardized in both groups and their correlation is set to $0.5$ in group 1 and to $0.5 + \textit{true difference}$ in group 2. We assume that interest focuses on this true difference in latent correlation between groups.

Multiple group structural equation modeling is compared with IPC regression on a group indicator as described above. 
The two procedures are evaluated under conditions determined by fully crossing the following simulation experimental factors:

\begin{itemize}
	\item The true difference in between-group correlation: \{-0.4, -0.2, -0.1, 0, +0.1, 0.2, 0.4\}. Note that the zero condition corresponds to no effect;
	\item The difference in loadings between groups: \{0, 0.1, 0.2\};
	\item The sample size per group: $n_g = $ \{125, 250, 500, 1000\}. The total sample size therefore equaled $n = 2 n_g = $ \{250, 500, 1000, 2000\}. 
\end{itemize}

For each of the $4 \times 7 \times  3 = 84$ resulting conditions, $n$ observations were drawn from a multivariate normal distribution with mean vector and covariance matrix depending on the group. The two procedures, IPC differences and MG-SEM, were then applied to the sample data, and this process was replicated 1000 times for each condition. 

\subsection{Results}

Figure \ref{fig:sim-bias-ipc} plots the true difference in latent correlations in each condition against the difference estimated using the differences in individual parameter contributions. For easy reference, the black 45-degree diagonal line corresponds to exact equality of the true and estimated difference. Each point corresponds to a particular simulation condition, connected with lines. The lines lie very close to the ideal of exact equality. When the true difference is strongly positive or negative, it is slightly underestimated in absolute terms. That is, the difference estimate is slightly biased towards zero for very large between group differences in the latent correlation.  The differences between the various conditions obtained from crossing sample size with loading differences are hardly discernible. This implies that the IPC difference procedure provides close to unbiased estimates under all of these conditions.
\begin{figure}\centering
	\includegraphics[width=.8\textwidth]{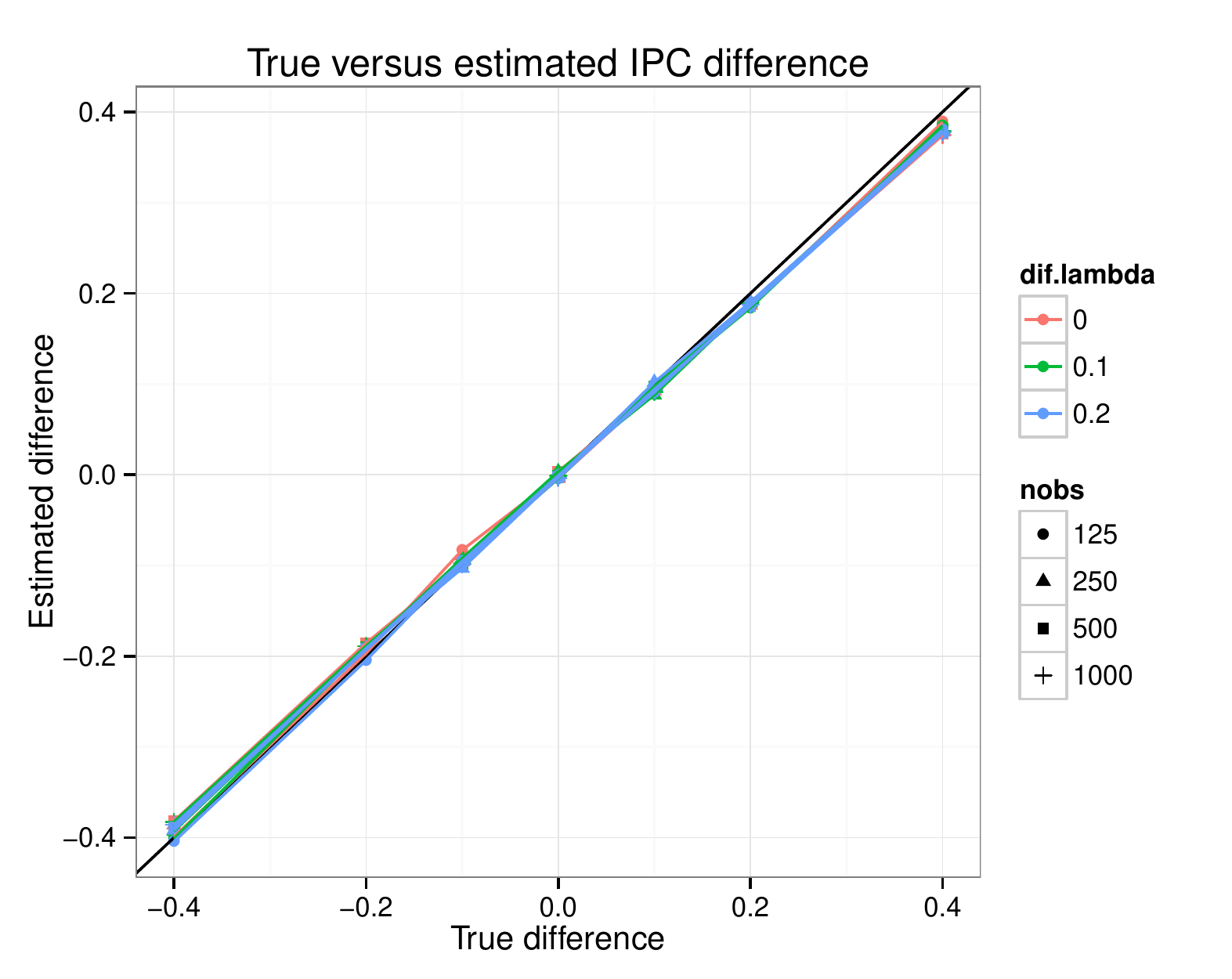}
	\caption{Bias in the mean difference between groups of IPC for the latent correlation. 
		Different conditions are shown as differences in color 
	(loading non-invariance) and point shape (number of sample observations).}
		\label{fig:sim-bias-ipc}
\end{figure}

Figure \ref{fig:sim-bias-mgsem} shows plots the same quantities as Figure \ref{fig:sim-bias-ipc} when using multiple group structural equation modeling (MG-SEM). The right-hand side displays the results obtained when estimating a correctly specified multiple group structural equation model. In this simulation setup, that implies there should be no misspecified equality restrictions across groups in the multiple group model. As clearly seen in Figure \ref{fig:sim-bias-mgsem}, the correctly specified multiple group model provides unbiased estimates under all conditions, including those with large differences in latent correlation. The correctly specified MG-SEM procedure is therefore  superior to the IPC difference procedure, which showed a slight bias towards zero in the extreme conditions.

The left-hand side of Figure \ref{fig:sim-bias-mgsem} demonstrates the bias that occurs when the covariance structure parameters are erroneously constrained to equality across the groups (the distance from the diagonal to the colored lines). This bias can be considerable. Moreover, when the true difference is zero, the misspecified equality restrictions yield a positive latent correlation difference estimate. In other words, the researcher will, on average, find a difference in latent correlations when there is none; at the same time, when there is a true difference in latent correlations, this difference may be severely underestimated. The difference between Figure \ref{fig:sim-bias-ipc} and the left-hand graph of Figure is striking: after all, the IPC difference procedure is based only on the results of a one-group pooled model, which  implicitly also restricts all covariance structure parameters to equality over groups. Figure \ref{fig:sim-bias-mgsem} shows that misspecifications in between-group equality restrictions are less influential in the IPC procedure than in multiple group modeling.

\begin{figure}\centering
	\includegraphics[width=\textwidth]{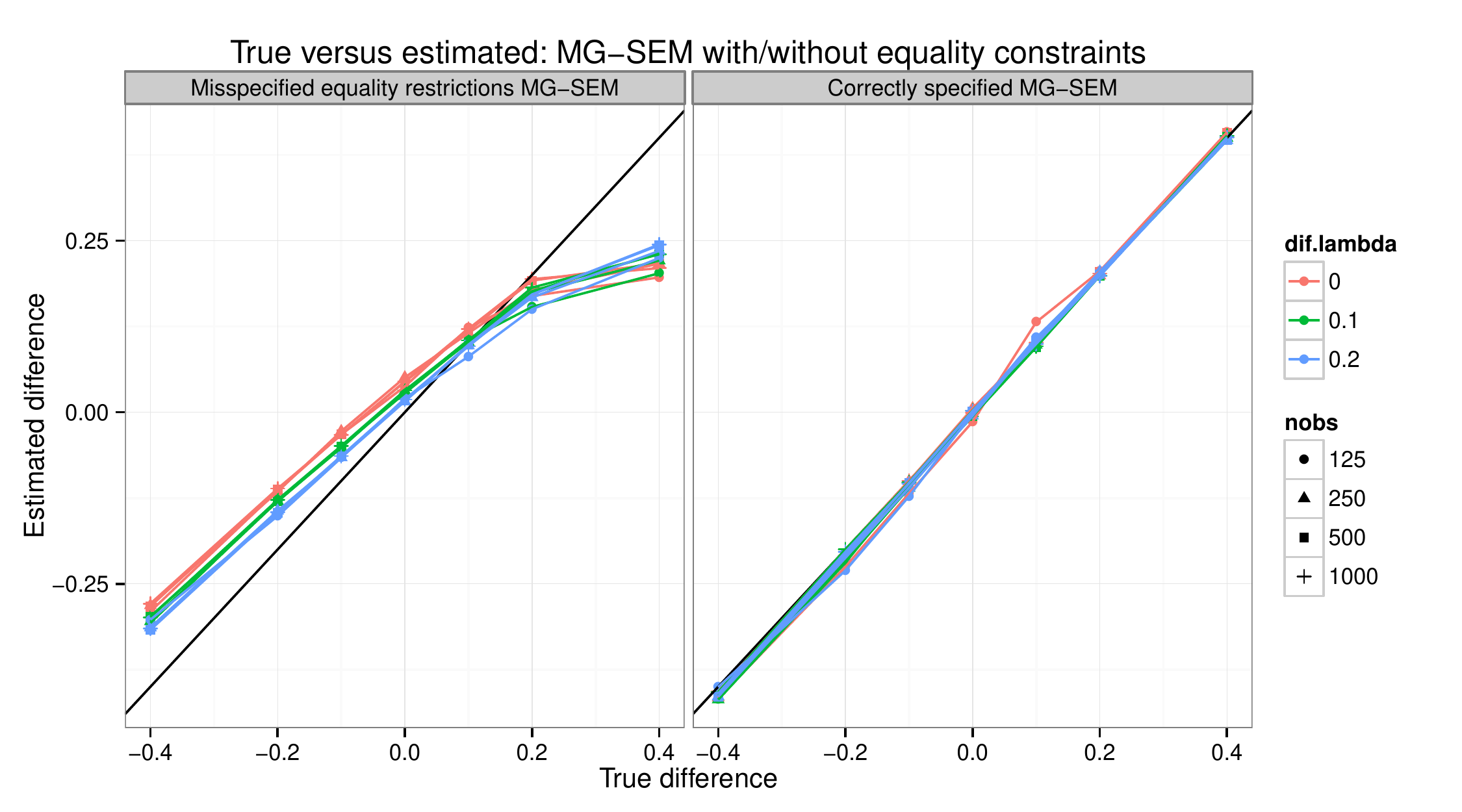}
	\caption{Bias in latent correlation differences obtained from 
	multiple group SEM with misspecified equality constraints (left) and correctly specified MG-SEM (right). 
}
		\label{fig:sim-bias-mgsem}
\end{figure}

One may also wonder whether, when we simply calculate the difference in IPC scores between groups, the usual standard error of this difference is valid and 95\% confidence  intervals so calculated actually cover the true difference in 95\% of the replications.
Figure \ref{fig:sim-coverage} shows that this is indeed the case. The dotted reference line indicates the ideal, 95\% coverage. Under all conditions, the confidence interval for the difference in IPC's produced by standard software covers the true difference in latent correlations about 95\% of the time. 

\begin{figure}
\centering
	\includegraphics[width=.8\textwidth]{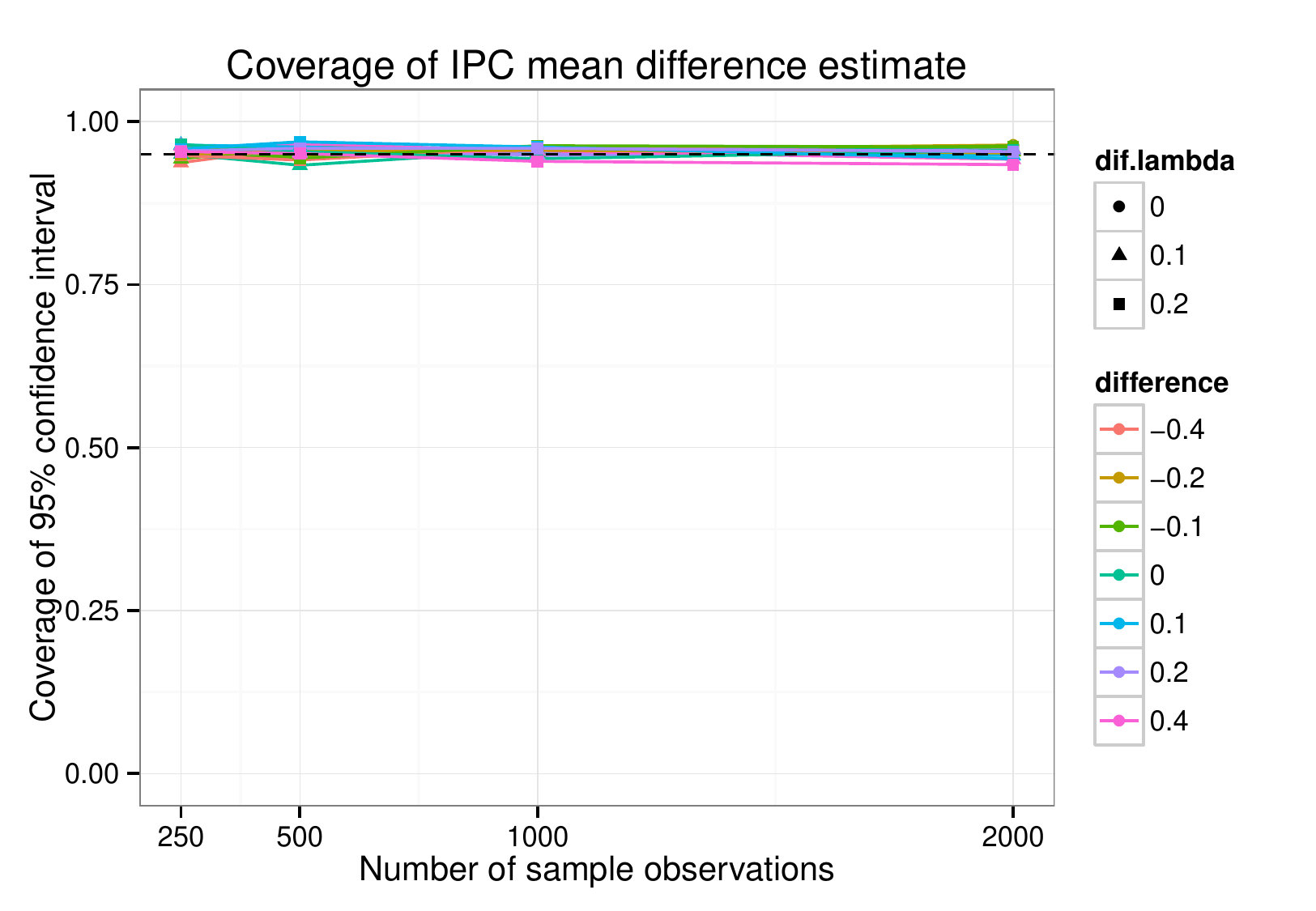}
	\caption{Coverage of 95\% confidence intervals of the mean difference between groups of latent correlation's IPC's. 
	Different conditions are shown as differences in color 
	(true latent correlation difference) and point shape (loading difference).}
		\label{fig:sim-coverage}
\end{figure}

Next, we examine the performance of the test of the hypothesis that the true between-group difference in latent correlation is zero.

Table \ref{tab:sim-alpha} shows, for both methods, the number of times the sample hypothesis test of equal latent correlations was rejected under the condition \emph{true difference} = 0: i.e., the empirical Type-I error rates using $\alpha = 0.05$. Under the null hypothesis, this proportion of rejected sample tests should be about 0.05. It can be seen that both procedures provide adequate control of type-I error under the null hypothesis. This holds true  for the smaller sample sizes as well as when other parameters do differ strongly over the groups.

\begin{table}[tb]
\centering
\caption{Empirical Type-I error rates for $\alpha = 0.05$.}\label{tab:sim-alpha}
\begin{tabular}{llrrrrrrr}
  \hline
  &&\multicolumn{3}{c}{IPC} &&\multicolumn{3}{c}{MG-SEM}\\
  \cline{3-5}\cline{7-9}
& \emph{Diff. loading} & 0 & 0.1 & 0.2 & & 0 & 0.1 & 0.2 \\ 
  \hline
\multicolumn{2}{l}{\emph{No. obs.}}\\  
&125 & 0.048 & 0.050 & 0.035	&& 0.043 & 0.050 & 0.053 \\ 
&  250 & 0.050 & 0.067 & 0.046	&& 0.035 & 0.040 & 0.045 \\ 
&  500 & 0.049 & 0.044 & 0.057	&& 0.049 & 0.059 & 0.053 \\ 
&  1000 & 0.048 & 0.043 & 0.044	&& 0.058 & 0.048 & 0.053\\ 
   \hline
\end{tabular}
\end{table}

When the \emph{true difference} $\neq$ 0, the rejection rate is the empirical power of the test to detect a difference in latent correlation. Figure \ref{fig:sim-power} relates the power of the test to sample size and the true size of the difference. The left-hand graph in Figure \ref{fig:sim-power} corresponds to the power using MG-SEM, and the right-hand graph to the power using the IPC difference test. The overall median power using multiple group structural equation modeling is 0.0095 higher than that using IPC difference testing. For conditions with a small and medium absolute differences (bottom lines), MG-SEM is somewhat more powerful, but with large differences (top lines),  IPC difference testing was found more powerful for smaller sample sizes.

\begin{figure}
\centering
	\includegraphics[width=\textwidth]{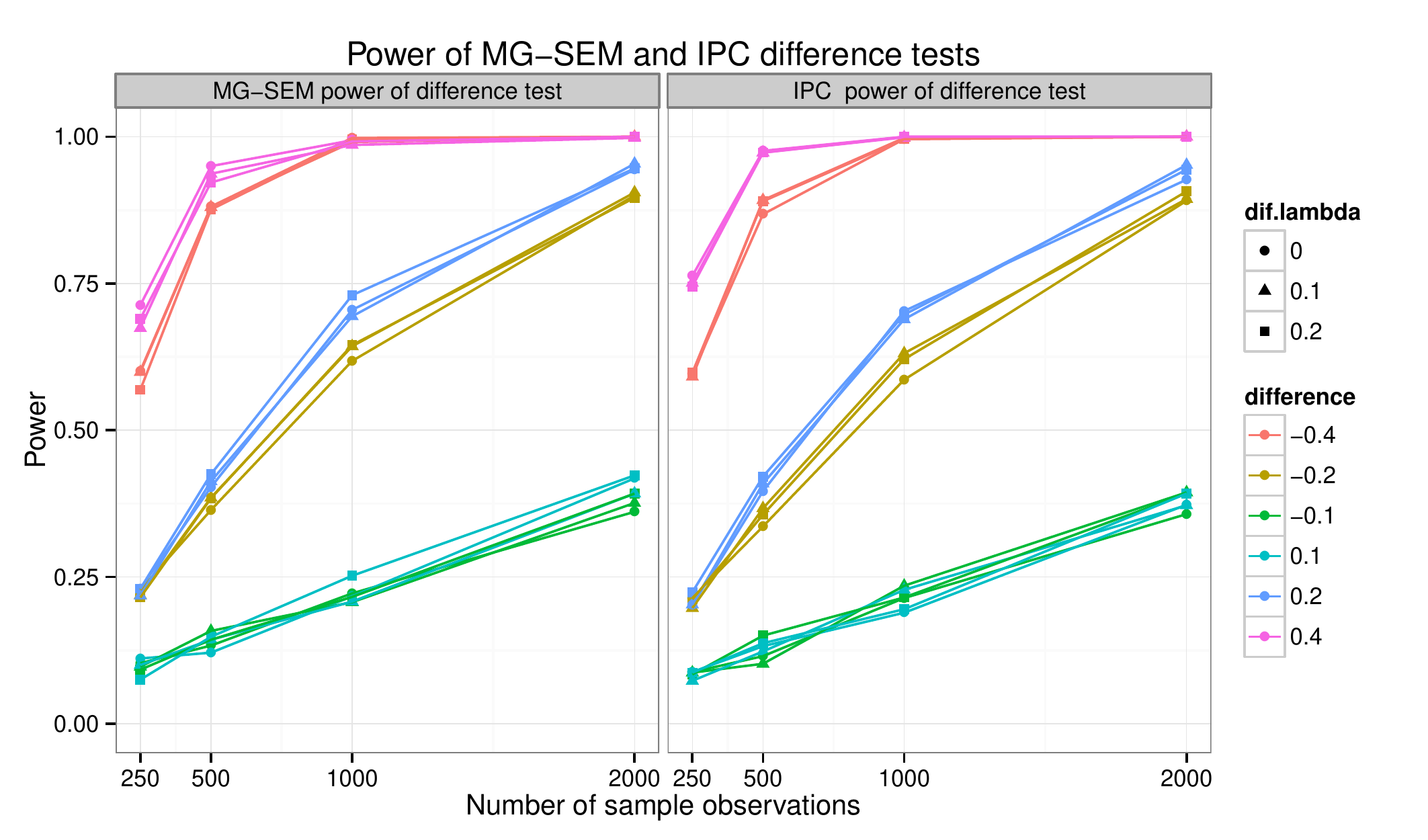}
	\caption{Power to reject mean difference test for latent correlation. Different conditions are shown as differences in color 
	(true latent correlation difference between groups) and point shape (loading non-invariance).}
	\label{fig:sim-power}
\end{figure}

\vspace{12pt}
Overall, the three-step individual parameter change difference procedure evaluated here performed almost as well as correctly specified multiple group structural equation modeling. The advantages of the IPC difference procedure in the simple setup discussed here that it is much simpler to apply, and that it may be more robust to misspecified equality restrictions than multiple group SEM.

\section{Relating measurement error variance in a survey question to respondent characteristics}\label{sec:application}

Internet usage is often studied in relationship with other variables, for instance in the literature on pathological internet use \citep[e.g.][]{gross2004adolescent} and in studies on social inequalities in internet use \citep[e.g.][]{hoffman1998bridging}.
When studying the relationship between internet use and other variables, random measurement error introduced by the self-reports may bias the results downwards \citep{lord_statistical_1968,fuller1987}. Worse, when random measurement error variance differs over groups of sex, age, and education, true differences may not only be obscured, but spurious differences may be erroneously found \citep{carroll2006measurement}. Therefore important questions are 1) how much measurement error variance the answers to internet use questions contain, and 2) whether this measurement error variance differs over social groups. We demonstrate how these questions may be answered using individual parameter contributions and contribute to the literature on survey errors.

We use 2838 observations from the LISS panel study, a simple random sample of Dutch households. For more information on the design of the LISS study we refer to \citet{scherpenzeel2011data}, and for the freely available original data and questionnaires to \url{http://lissdata.nl/}. In this longitudinal panel study design, four consecutive measurements of internet usage at work are obtained in the years 2008--2011. The log-transform of the number of hours per week the respondent claims to use the internet at work is taken to reduce skewness; Table \ref{tab:application-descriptives} shows descriptive statistics for these measures. 

\begin{table}[tb]
\centering
\caption{Descriptive statistics for four years' observations of 
	log(hours of internet use at work per week + 1); $n = 2838$. Covariates not shown.}
	\label{tab:application-descriptives}

\begin{tabular}{rrrrrrrrr}
  \hline
&  \multicolumn{4}{c}{Correlations}\\
  \cline{2-5}
 & 2008 & 2009 & 2010 & 2011 & Mean & sd & Skew & Kurtosis \\ 
  \hline
inet\_2008 & 1 & && & 0.721 & 1.029 & 1.298 & 0.600 \\ 
  inet\_2009 & 0.668 & 1 & & & 0.743 & 1.048 & 1.283 & 0.543 \\ 
  inet\_2010 & 0.644 & 0.701 & 1 &  & 0.766 & 1.073 & 1.229 & 0.315 \\ 
  inet\_2011 & 0.609 & 0.661 & 0.729 & 1 & 0.790 & 1.094 & 1.181 & 0.189 \\    
  \hline
\end{tabular}
\end{table}

\begin{figure}[tb]
\centering
	\includegraphics[width=\textwidth]{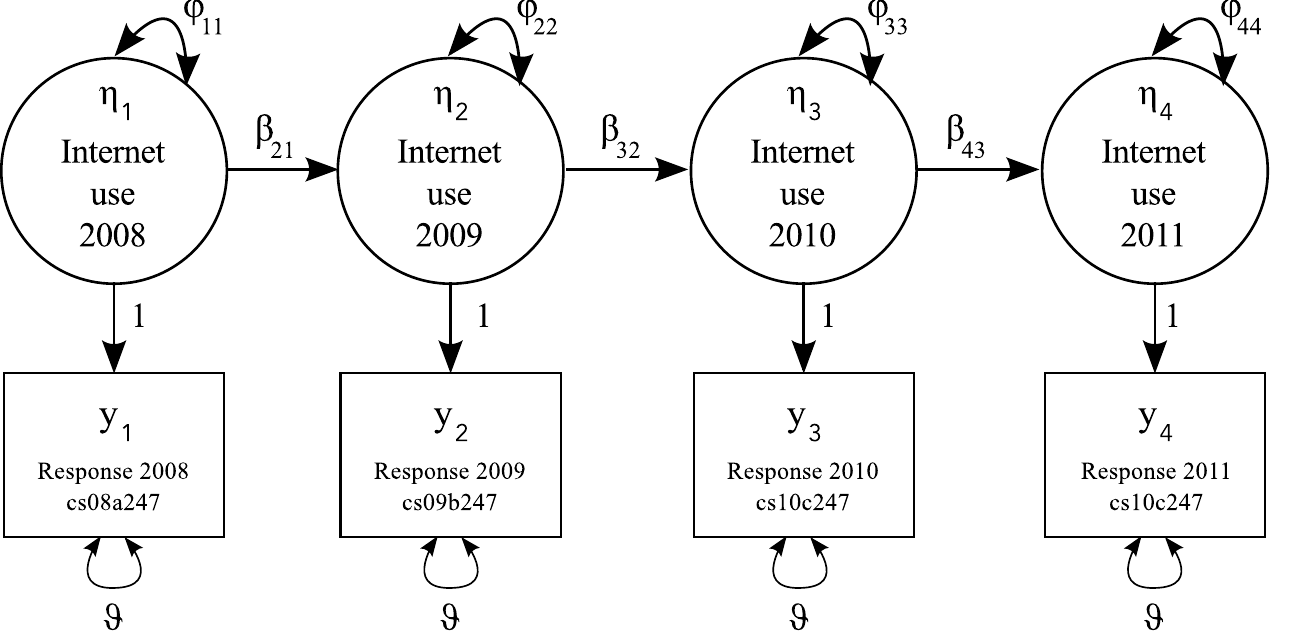}
	\caption{The quasi-simplex model for four consecutive measurement error-prone measurements of internet
	use in the LISS panel.}
		\label{fig:quasi-simplex}
\end{figure}

The first step is to formulate a structural equation model estimating measurement error variance in the answers to the internet use question. Table \ref{tab:application-descriptives} shows that the consecutive measurements do not all correlate equally: measurements at a greater distance from each other in time correlate less, so that the size of the correlations tapers off towards the lower-left corner of the matrix. This is in line with the so-called ``quasi-simplex'' model, in which internet use is not only measured with error, but the true internet use may also change over time \citep{wiley35wiley,alwin_reliability_1991,alwin_margins_2007}. Figure \ref{fig:quasi-simplex} shows the single-group quasi-simplex model formulated for these data as a structural equation model. 

\begin{table}[tb]
\begin{center}
\caption{Results of fitting the quasi-simplex model with covariates to the LISS data.
	Intercepts and fixed covariate main effects not shown.
	Note: Satorra-Bentler chi-square is 1.5 with 2 degrees of freedom ($p=0.476$, scaling factor = 2.3). CFI = 1.000, RMSEA = 0.016, $n = 2838$.	}
	\label{tab:application-estimates}
\begin{tabular}{rrrrrrrrrrrrr}
  \hline
  && \multicolumn{3}{c}{Latent regressions} & & \multicolumn{4}{c}{Variances}  \\ 
  \cline{3-5}\cline{7-10}
 & var(Error) & $\beta_{21}$ & $\beta_{32}$ & $\beta_{43}$ && F1 & F2  & F3 & F4  \\ 
  \hline
Est. & 0.28 & 0.91 & 0.94 & 0.94 & &0.58 & 0.13 & 0.11 & 0.08 \\ 
  s.e. & (0.01) & (0.03) & (0.02) & (0.02) && (0.02) & (0.02) & (0.01) & (0.02) \\ 
Stand. & 0.26 & 0.89 & 0.90 & 0.92 && 0.74 & 0.16 & 0.12 & 0.09 \\ 
   \hline
\end{tabular}

\end{center}
\end{table}

Fitting this model to the LISS data with the ``robust'' maximum likelihood estimator \citep{satorra1994corrections} yields a well-fitting model, with estimates shown in Table \ref{tab:application-estimates}. The model includes the direct effect of self-employment, age, age squared, sex, and education level on each of the four observed variables. These regression coefficient estimates are omitted from Table \ref{tab:application-estimates} for clarity.
The error variance parameter is estimated at $0.28 \pm 0.01$. This error variance estimate corresponds to reliability estimates of respectively $0.86$, $0.86$, $0.87$, and $0.88$ at the four consecutive years 2008 through 2011. The overall reliability is therefore high, but not perfect: estimates of the relationship between internet use and social background will be affected.

Another question of interest is whether the measurement error variance differs over groups. Differential measurement error can strongly affect relationship estimates \citep{carroll2006measurement}. Furthermore, survey methodologists want to know what types of respondents provide higher or lower quality data \citep[e.g.][]{alwin_reliability_1991,scherpenzeel1995question,alwin_margins_2007,revilla2012impact}. Having completed step one of fitting a single-group structural equation model, we now obtain each individual's parameter contribution to the error variance parameter estimate (step two):  a $\hat{\IPC}$ score for each respondent. Step three is then to perform a multiple linear regression of that score on respondent characteristics using standard routines. We employ the \texttt{lm} function in \texttt{R}, but linear regression routines such as those in Excel, SPSS, or Stata would serve equally well. The characteristics investigated here are whether the respondent is self-employed or not (0/1), the year of birth and its square, sex of the respondent, and education level -- primary, lower secondary (VMBO), middle secondary (MBO), upper secondary (HAVO/VWO), lower tertiary (HBO), and higher tertiary (University).

\begin{figure}[tb]
	\includegraphics[width=\textwidth]{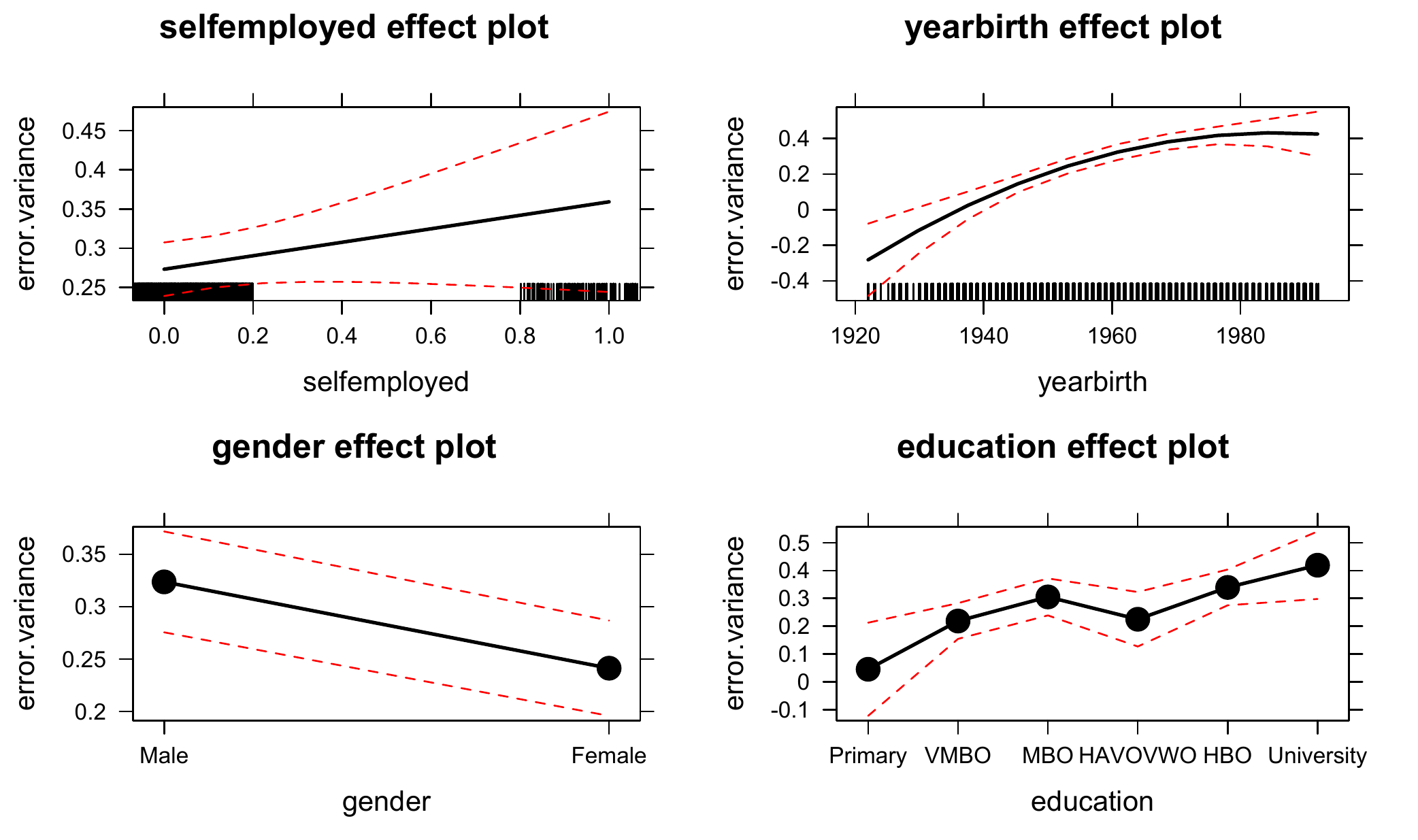}
	\caption{
	Estimated conditional differences in measurement error variance by respondent characteristics: effect plot for
  the regression of the IPC of the error variance parameter on 
	self-employed, year of birth, gender, and education.}
	\label{fig:application-effectsplot}
\end{figure}

The multiple linear regression estimates of the relationship between respondent characteristics and error variance are shown as effect plots in Figure \ref{fig:application-effectsplot}. The estimates themselves can be found in the appendix. Figure \ref{fig:application-effectsplot} shows that, controlling for the other characteristics, three of the four respondent characteristics (age, sex, and education) have a strong and statistically significant effect on the error variance, and therefore on reliability, which is inversely related to the error variance. The reliability of internet use self-reports \emph{increases} with age, \emph{decreases} with education, and is lower for men than for women. 

To survey methodologists and those familiar with the literature on reliability of survey questions, it may seem surprising that reliability would be lower for the young and those of a higher education level. After all, cognitive functioning is higher for these groups and expected to increase the reliability of answers \citep{krosnick2010experiments}. 
There may, however, be a simple explanation for this finding: the variance of measurement errors in self-reports of the number of hours spent using the internet at work may simply be related to internet use itself. If one does not use the internet \emph{at all}, presumably there would be little to no measurement error. A possible explanation for the finding that young and educated respondents appear to provide lower-quality answers is that these respondents simply tend not to be the kind of respondents who use the internet at work for zero hours per week. 

The statistically significant effects on error variance of these respondent characteristics may therefore be \emph{mediated} by whether or not the respondent \emph{never} uses the internet at work. As an indicator of this variable, we construct a new dummy variable that equals one if the respondent has answered ``zero hours'' on all four occasions and zero otherwise. We then formulate a new structural equation full mediation model in which the error variance IPC is influenced by the ``all-zeroes'' variable, and this variable, in turn is influenced by the four covariates. This full mediation model fits the data well ($\chi_{\textrm{SB}}^2 = 11$, df = 9, $p = 0.276$, RMSEA = 0.009; full results can be found in the appendix). Thus, it appears that measurement error variance indeed differs over these respondent characteristics, but only for the reason that respondents with these characteristics tend not to be those who never use the internet, something that tends to strongly reduce the error variance: the difference is $0.50 \pm 0.03$ corresponding to a fully standardized effect of ``all-zeroes'' on the error variance of 0.3.

An evaluation of the IPC regression results can be made by comparing the IPC mean differences with the results that would have been obtained by separate multiple group analyses. 
It is not possible to formulate a multiple group structural equation model in which the effect of some grouping variables on the error variance are controlled for the effect of other grouping variables. However, it is possible to repeat the IPC regression (the third step), but this time regressing the IPC on only one covariate at a time. Since there are four covariates, this yields four linear IPC regressions. This analysis can then be compared with four separate multiple group analyses in which the groups are formed by each of the four covariates (no between-group restrictions are imposed). We may then compare the differences in the error variance estimates over groups obtained from each of the four multiple group analyses with those obtained from the four simple IPC regressions. 

\begin{table}[tb]\centering
	\caption{Second column: estimates of between-group difference in error variance parameter obtained by taking a difference in IPC's; third column: difference estimated by fitting a multiple group SEM in each group and calculating the difference in the maximum-likelihood estimates.}\label{tab:application-bias}
	\begin{tabular}{llrrr}
	\hline 
& & IPC regression & Multiple group SEM & Bias\\
\hline
\multicolumn{2}{l}{Female $-$ male} & 0.0576 & 0.0627 & -0.0051\\
\multicolumn{2}{l}{Self-employed $-$ not} & 0.0919 & 0.0939 & -0.0020\\
\multicolumn{2}{l}{\emph{Education}  $-$ Primary} \\
&VMBO      & 0.1575 & 0.1638 & -0.0064\\
&MBO  & 0.3122 & 0.3196 & -0.0051\\
&HAVO/VWO    & 0.2271 & 0.2219 & -0.0051\\
&HBO & 0.3141 & 0.3225 & -0.0051\\
&WO & 0.4133 & 0.4420 & -0.0051\\
\hline 
	\end{tabular}
\end{table}

Table \ref{tab:application-bias} shows that the estimated differences in error variance between the groups obtained by either multiple group structural equation modeling or regression of the individual parameter contributions on a covariate are very close to each other. Since multiple group SEM is an asymptotically unbiased procedure, we have labeled the difference between the two sets of estimates ``bias'' (fourth column of Table \ref{tab:application-bias}). Overall, the bias is small. Thus, in this example, if a single grouping variable is of interest, the results obtained by multiple group SEM are almost identical to those obtained by simply assessing between-group differences in IPC's. Multiple group SEM does not, however, allow for multiple grouping variables whose shared variance is to be partialed out.

\section{Discussion and conclusion}

We suggested to study differences in structural equation model parameters over groups by a simplified three-step procedure. In the first step a single-group SEM is fit to the data; in step two the transformation discussed in section \ref{sec:theory} is applied to the data to obtain individual parameter contributions (IPC's); and in step three, these IPC's are related to grouping variables by standard methods such as dummy regression or $t$-tests. More generally, we showed that performing a linear regression of the IPC's on a set of covariates, categorical or continuous, amounts to calculating misspecification-robust ``expected parameter changes'' and ``modification indices'' for the heterogeneity of an arbitrary SEM parameter estimate with respect to the covariates. The procedure is simple to apply, since after step two it amounts simply to regressing a dependent variable, the IPC for a parameter, on a set of covariates explaining heterogeneity in that parameter.

A small Monte Carlo simulation evaluated the performance of this method in a particular setup, namely when interest focuses on between-group differences in latent factor correlations in a two-factor model. The IPC regression method performed excellently in this setup, showing little bias and good coverage properties. Moreover, we compared the IPC regression method in this setup to the more common, and theoretically most efficient, multiple group SEM. There were few differences between the two methods in bias, variance, type I error control, or power. In our application of the IPC regression method to between-group variations in error variances, no differences larger than 8\% between the IPC regression method and multiple group SEM were found.

Naturally, there are many possible structural equation models, parameters of interest, and heterogeneity setups; no single simulation study can evaluate the performance of this method under all of them, nor was that  our goal. Still, this represents a limitation of the present simulation study.  Our theoretical results as well as the correspondence between the empirical results in the application to the quasi-simplex model suggest that the method may show good properties in other situations as well, although a larger range of situations remains to be investigated. Exceptions may occur when the base single-group model is extremely misspecified, something that may occur in certain models (e.g. growth models) when possible mean differences with respect to the covariates are not included in the single-group model. For this reason we advise including the covariates' main effects in the single-group analysis before estimating the individual parameter contributions.

Using four consecutive waves from a longitudinal panel study of Dutch households, we estimated measurement error variance in self-reports of internet use. Individual parameter contributions to the measurement error variance varied over age, gender, and education level. Contrary to what might be expected, however, young people and the highly educated tended to provide \emph{lower} quality information. We hypothesized an explanation for this counterintuitive finding: that older people, the lower-educated, and women tend more often to simply \emph{never} use the internet, a true value that would not easily incite response errors. A mediation model in which the variable  ``answering zero hours on all four measurement occasions'' mediated the effect of the respondent characteristics on individual parameter contributions to the error variance was in line with that hypothesis. After taking into account that those who always use the internet for zero hours per week have a lower measurement error variance, there was no further effect of respondent characteristics on the error variance. 
The finding that measurement quality in a frequency question is related to the frequency itself corresponds to findings by \citet{tourangeau2000psychology} and \citet{saris_quality_2008}. This finding could not be explained by differences in respondent characteristics. Moreover, our study shows that, because of this effect, there is differential measurement error in internet use self-reports with respect to age, gender, and education level. Therefore researchers explaining or predicting internet use should be careful to allow for this differential measurement error so as to prevent bias.

Extensions to the suggested use of IPC's are trivial but require evaluation as to their efficacy. For example, clustering of IPC's may be an attractive robust method of modeling heterogeneity in parameters without allowing the clusters to be determined by heteroskedasticity and mean differences. The method may also be used to investigate measurement invariance, in which case a comparison with the test statistics suggested by  \citet{merkle2013tests} would be an interesting avenue of further study. It may be possible to model derived parameters, such as standardized coefficients and indirect or total effects in mediation analysis. Categorical data and ``extended structural equation models'' such as latent class analysis were not discussed here; both could in principle be accomodated by adopting a framework such as that of \citet{muthen2002beyond} or \citet{dutoit2012analysis}. Finally, we did not deal with missing data. Here, one avenue would be casewise gradients with respect to the FIML fitting function, and another approach might be to simply delete missings ``pairwise'' in the multiplication $\W \cdot \d$. How these approaches perform in practice, and how they perform relative to the currently best available methods, remains to be investigated.

\vspace{12pt}
In summary, individuals may differ in their parameter values. This article discussed a considerable simplification and extension of the study of arbitrary SEM parameter heterogeneity with respect to covariates. R code implementing this idea for the SEM package lavaan is given in the Appendix, as well as an example application. Using individual parameter contributions, we were able to provide a contribution to survey methodology that would have been difficult to make otherwise. We hope that other structural equation modelers will also find the suggested method a useful tool to study individual differences in SEM parameters.

%%%%%%%%%%%%%%%%%%%% REFERENCES %%%%%%%%%%%%%%%%%%%%%%%
\bibliography{/Users/daob/Dropbox/Bibliography/quality.bib}

\begin{thebibliography}{}

\bibitem[Alwin, 2007]{alwin_margins_2007}
Alwin, D. (2007).
\newblock {\em Margins of error: a study of reliability in survey measurement}.
\newblock {Wiley-Interscience}, New York.

\bibitem[Alwin and Krosnick, 1991]{alwin_reliability_1991}
Alwin, D. and Krosnick, J. (1991).
\newblock The reliability of survey attitude measurement: The influence of
  question and respondent attributes.
\newblock {\em Sociological Methods \& Research}, 20:139.

\bibitem[Andrews, 1993]{andrews1993tests}
Andrews, D.~W. (1993).
\newblock Tests for parameter instability and structural change with unknown
  change point.
\newblock {\em Econometrica: Journal of the Econometric Society}, pages
  821--856.

\bibitem[Bentler and Chou, 1992]{bentler1992some}
Bentler, P. and Chou, C. (1992).
\newblock Some new covariance structure model improvement statistics.
\newblock {\em Sociological Methods \& Research}, 21(2):259--282.

\bibitem[Bentler and Dijkstra, 1984]{bentler1984efficient}
Bentler, P. and Dijkstra, T. (1984).
\newblock Efficient estimation via linearization in structural models.
\newblock In Krishnaiah, P., editor, {\em Multivariate Analysis {VI}}, pages
  9--42. North-Holland, Amsterdam.

\bibitem[Bollen, 1989]{bollen_structural_1989}
Bollen, K. (1989).
\newblock {\em Structural Equations with Latent Variables}.
\newblock John Wiley \& Sons, New York.

\bibitem[Bollen, 1995]{bollen1995structural}
Bollen, K. (1995).
\newblock Structural equation models that are nonlinear in latent variables: A
  least-squares estimator.
\newblock {\em Sociological methodology}, 25:223--252.

\bibitem[Boos, 1992]{boos1992generalized}
Boos, D. (1992).
\newblock On generalized score tests.
\newblock {\em The American Statistician}, 46(4):327--333.

\bibitem[Brown et~al., 1975]{brown1975techniques}
Brown, R.~L., Durbin, J., and Evans, J.~M. (1975).
\newblock Techniques for testing the constancy of regression relationships over
  time.
\newblock {\em Journal of the Royal Statistical Society. Series B
  (Methodological)}, pages 149--192.

\bibitem[Carroll et~al., 2006]{carroll2006measurement}
Carroll, R., Ruppert, D., Stefanski, L., and Crainiceanu, C. (2006).
\newblock {\em Measurement error in nonlinear models: a modern perspective},
  volume 105.
\newblock Chapman \& Hall/CRC Monographs on Statistics \& Applied Probability,
  Boca Raton, FL.

\bibitem[Davidov et~al., 2008]{davidov2008values}
Davidov, E., Meuleman, B., Billiet, J., and Schmidt, P. (2008).
\newblock Values and support for immigration: a cross-country comparison.
\newblock {\em European Sociological Review}, 24(5):583--599.

\bibitem[{du Toit}, 2012]{dutoit2012analysis}
{du Toit}, S. (2012).
\newblock Analysis of structural equation models based on a mixture of
  continuous and ordinal random variables in the case of complex survey data.
\newblock In Edwards, M.~C. and MacCallum, R.~C., editors, {\em Current Issues
  in the Theory and Application of Latent Variable Models}. Routledge Academic,
  New York.

\bibitem[Fuller, 1987]{fuller1987}
Fuller, W. (1987).
\newblock {\em Measurement Error Models}.
\newblock John Wiley \& Sons, New York.

\bibitem[Gross, 2004]{gross2004adolescent}
Gross, E.~F. (2004).
\newblock Adolescent internet use: What we expect, what teens report.
\newblock {\em Journal of Applied Developmental Psychology}, 25(6):633--649.

\bibitem[Hjort and Koning, 2002]{hjort2002tests}
Hjort, N.~L. and Koning, A. (2002).
\newblock Tests for constancy of model parameters over time.
\newblock {\em Journal of Nonparametric Statistics}, 14(1-2):113--132.

\bibitem[Hoffman and Novak, 1998]{hoffman1998bridging}
Hoffman, D. and Novak, T. (1998).
\newblock Bridging the racial divide on the internet.
\newblock {\em Science}, 280(5362):390--391.

\bibitem[Jaccard and Wan, 1996]{jaccard1996lisrel}
Jaccard, J. and Wan, C.~K. (1996).
\newblock {\em {LISREL} approaches to interaction effects in multiple
  regression}.
\newblock Sage, Thousand Oaks, CA.
\newblock Sage University Paper series on Quantitative Applications in the
  Social Sciences, series no. 114.

\bibitem[J{\"o}reskog and Yang, 1996]{joreskog1996nonlinear}
J{\"o}reskog, K.~G. and Yang, F. (1996).
\newblock Nonlinear structural equation models: The {Kenny-Judd} model with
  interaction effects.
\newblock In Marcoulides, G. and Schumacker, R., editors, {\em Advanced
  Structural Equation Modeling : Issues and Techniques}, pages 57--88. Lawrence
  Erlbaum Associates.

\bibitem[Kendler et~al., 2000]{kendler2000illicit}
Kendler, K.~S., Karkowski, L.~M., Neale, M.~C., and Prescott, C.~A. (2000).
\newblock Illicit psychoactive substance use, heavy use, abuse, and dependence
  in a us population-based sample of male twins.
\newblock {\em Archives of general psychiatry}, 57(3):261.

\bibitem[Kenny and Judd, 1984]{kenny1984estimating}
Kenny, D.~A. and Judd, C.~M. (1984).
\newblock Estimating the nonlinear and interactive effects of latent variables.
\newblock {\em Psychological Bulletin; Psychological Bulletin}, 96(1):201.

\bibitem[Klein and Moosbrugger, 2000]{klein2000maximum}
Klein, A. and Moosbrugger, H. (2000).
\newblock Maximum likelihood estimation of latent interaction effects with the
  lms method.
\newblock {\em Psychometrika}, 65(4):457--474.

\bibitem[Krosnick, 2011]{krosnick2010experiments}
Krosnick, J.~A. (2011).
\newblock Experiments for evaluating survey questions.
\newblock In Madans, J., Miller, K., Maitland, A., and Willis, G., editors,
  {\em Question Evaluation Methods: Contributing to the Science of Data
  Quality}, pages 213--238. Wiley Online Library, New York.

\bibitem[Kuan and Hornik, 1995]{kuan1995generalized}
Kuan, C.-M. and Hornik, K. (1995).
\newblock The generalized fluctuation test: A unifying view.
\newblock {\em Econometric Reviews}, 14(2):135--161.

\bibitem[Lord and Novick, 1968]{lord_statistical_1968}
Lord, F.~M. and Novick, M.~R. (1968).
\newblock {\em Statistical theories of mental scores}.
\newblock Addison--Wesley, Reading.

\bibitem[MacCallum and Austin, 2000]{maccallum2000applications}
MacCallum, R.~C. and Austin, J.~T. (2000).
\newblock Applications of structural equation modeling in psychological
  research.
\newblock {\em Annual review of psychology}, 51(1):201--226.

\bibitem[Merkle and Zeileis, 2013]{merkle2013tests}
Merkle, E.~C. and Zeileis, A. (2013).
\newblock Tests of measurement invariance without subgroups: A generalization
  of classical methods.
\newblock {\em Psychometrika}, 78(1):59--82.

\bibitem[Mooijaart and Bentler, 2010]{mooijaart2010alternative}
Mooijaart, A. and Bentler, P.~M. (2010).
\newblock An alternative approach for nonlinear latent variable models.
\newblock {\em Structural Equation Modeling}, 17(3):357--373.

\bibitem[Muth\'en, 1989]{muthen1989latent}
Muth\'en, B. (1989).
\newblock Latent variable modeling in heterogeneous populations.
\newblock {\em Psychometrika}, 54(4):557--585.

\bibitem[Muth\'en, 1994]{muthen1994multilevel}
Muth\'en, B. (1994).
\newblock Multilevel covariance structure analysis.
\newblock {\em Sociological methods \& research}, 22(3):376--398.

\bibitem[Muth\'en, 2002]{muthen2002beyond}
Muth\'en, B. (2002).
\newblock Beyond {SEM}: General latent variable modeling.
\newblock {\em Behaviormetrika}, 29:81--117.

\bibitem[Muth\'en and Asparouhov, 2012]{muthen2012bayesian}
Muth\'en, B. and Asparouhov, T. (2012).
\newblock Bayesian {SEM}: A more flexible representation of substantive theory.
\newblock {\em Psychological Methods}, 17(3):313--35.

\bibitem[Muth\'en and Satorra, 1995]{muthen_complex_1995}
Muth\'en, B. and Satorra, A. (1995).
\newblock Complex sample data in structural equation modeling.
\newblock {\em Sociological methodology}, 25:267--316.

\bibitem[Neudecker and Satorra, 1991]{neudecker1991linear}
Neudecker, H. and Satorra, A. (1991).
\newblock {Linear Structural Relations: Gradient and Hessian of the Fitting
  Function}.
\newblock {\em Statistics and Probability Letters}, 11:57--61.

\bibitem[Nyblom, 1989]{nyblom1989testing}
Nyblom, J. (1989).
\newblock Testing for the constancy of parameters over time.
\newblock {\em Journal of the American Statistical Association},
  84(405):223--230.

\bibitem[Preacher et~al., 2006]{preacher2006computational}
Preacher, K.~J., Curran, P.~J., and Bauer, D.~J. (2006).
\newblock Computational tools for probing interactions in multiple linear
  regression, multilevel modeling, and latent curve analysis.
\newblock {\em Journal of Educational and Behavioral Statistics},
  31(4):437--448.

\bibitem[{R Core Team}, 2012]{Rlanguage}
{R Core Team} (2012).
\newblock {\em {R}: A Language and Environment for Statistical Computing}.
\newblock {R} Foundation for Statistical Computing, Vienna, Austria.
\newblock {ISBN} 3-900051-07-0.

\bibitem[Raykov and Marcoulides, 2006]{raykov2006first}
Raykov, T. and Marcoulides, G.~A. (2006).
\newblock {\em A First Course in Structural Equation Modeling, 2nd ed.}
\newblock Erblaum, Mahwah, NJ.

\bibitem[R\'evilla, 2012]{revilla2012impact}
R\'evilla, M. (2012).
\newblock Impact of the mode of data collection on the quality of answers to
  survey questions depending on respondent characteristics.
\newblock {\em Bulletin de M{\'e}thodologie Sociologique}, 116:44--60.

\bibitem[Rosseel et~al., 2013]{rosseel2013lavaan}
Rosseel, Y., Oberski, D., Byrnes, J., Vanbrabant, L., and Savalei, V. (2013).
\newblock lavaan: Latent variable analysis.
\newblock [Software].
\newblock Available from \url{http://lavaan.ugent.be/}.

\bibitem[Saris et~al., 2008]{saris_quality_2008}
Saris, W., Coromina, L., and Oberski, D. (2008).
\newblock The quality of the measurement of interest in the political issues in
  the media in the {ESS}.
\newblock {\em ASK Research \& Methods}, (17):7.

\bibitem[Saris et~al., 1987]{saris_detection_1987}
Saris, W., Satorra, A., and S\"orbom, D. (1987).
\newblock The detection and correction of specification errors in structural
  equation models.
\newblock {\em Sociological Methodology}, 17:105--129.

\bibitem[Satorra, 1989]{satorra1989alternative}
Satorra, A. (1989).
\newblock Alternative test criteria in covariance structure analysis: A unified
  approach.
\newblock {\em Psychometrika}, 54(1):131--151.

\bibitem[Satorra, 1992]{satorra_asymptotic_1992}
Satorra, A. (1992).
\newblock Asymptotic robust inferences in the analysis of mean and covariance
  structures.
\newblock {\em Sociological Methodology}, 22:249--278.

\bibitem[Satorra and Bentler, 1994]{satorra1994corrections}
Satorra, A. and Bentler, P. (1994).
\newblock Corrections to test statistics and standard errors in covariance
  structure analysis.
\newblock In von Eye, A. and Clogg, C.~C., editors, {\em Latent Variables
  analysis: applications to developmental research}. Sage, Thousand Oakes, CA.

\bibitem[Scherpenzeel, 1995]{scherpenzeel1995question}
Scherpenzeel, A. (1995).
\newblock {\em A question of quality. {Evaluating} survey questions by
  multitrait-multimethod studies}.
\newblock Royal PTT Nederland NV, Amsterdam.

\bibitem[Scherpenzeel, 2011]{scherpenzeel2011data}
Scherpenzeel, A. (2011).
\newblock Data collection in a probability-based internet panel: How the {LISS}
  panel was built and how it can be used.
\newblock {\em Bulletin of Sociological Methodology/Bulletin de
  M{\'e}thodologie Sociologique}, 109(1):56--61.

\bibitem[Scherpenzeel and Saris, 1997]{scherpenzeel_validity_1997}
Scherpenzeel, A. and Saris, W. (1997).
\newblock The validity and reliability of survey questions: {A} meta-analysis
  of {MTMM} studies.
\newblock {\em Sociological Methods \& Research}, 25:341.

\bibitem[Schumacker and Marcoulides, 1998]{schumacker1998interaction}
Schumacker, R.~E. and Marcoulides, G.~A. (1998).
\newblock {\em Interaction and nonlinear effects in structural equation
  modeling}.
\newblock Lawrence Erlbaum.

\bibitem[S{\"o}rbom, 1974]{sorbom1974general}
S{\"o}rbom, D. (1974).
\newblock A general method for studying differences in factor means and factor
  structure between groups.
\newblock {\em British Journal of Mathematical and Statistical Psychology},
  27(2):229--239.

\bibitem[S{\"o}rbom, 1989]{sorbom1989model}
S{\"o}rbom, D. (1989).
\newblock Model modification.
\newblock {\em Psychometrika}, 54(3):371--384.

\bibitem[Tourangeau et~al., 2000]{tourangeau2000psychology}
Tourangeau, R., Rips, L., and Rasinski, K. (2000).
\newblock {\em {The psychology of survey response}}.
\newblock Cambridge Univ Press, Cambridge, United Kingdom.

\bibitem[van Smeden and Hessen, 2013]{van2013testing}
van Smeden, M. and Hessen, D.~J. (2013).
\newblock Testing for two-way interactions in the multigroup common factor
  model.
\newblock {\em Structural Equation Modeling: A Multidisciplinary Journal},
  20(1):98--107.

\bibitem[Wiley and Wiley, 1970]{wiley35wiley}
Wiley, D. and Wiley, J.~A. (1970).
\newblock The estimation of measurement error in panel data.
\newblock {\em American Sociological Review}, 35(1):112--117.

\bibitem[Zeileis, 2005]{zeileis2005unified}
Zeileis, A. (2005).
\newblock A unified approach to structural change tests based on ml scores, f
  statistics, and ols residuals.
\newblock {\em Econometric Reviews}, 24(4):445--466.

\bibitem[Zeileis and Hornik, 2007]{zeileis2007generalized}
Zeileis, A. and Hornik, K. (2007).
\newblock Generalized {M-fluctuation} tests for parameter instability.
\newblock {\em Statistica Neerlandica}, 61(4):488--508.

\end{thebibliography}
%%%%%%%%%%%%%%%%%%%%%%%%%%%%%%%%%%%%%%%%%%%%%%%%%%%%%%%

\appendix

\section{A difference in IPC's and the corresponding test can be interpreted as an EPC and MI in a multiple group structural equation model}\label{sec:EPC}

The expected parameter change for freeing an equality-constrained parameter $\theta_g$ in group $g$ is defined as \citep{saris_detection_1987}:
\begin{equation}
	\EPC_g = \delta_g := \left( \frac{\partial^2 F}{\partial \boldsymbol{\theta}_g \partial \boldsymbol{\theta}_g'} \right)^{-1} \left( \frac{\partial F}{\theta_{g}  } \right)
\end{equation}
Note that the hessian with respect to the entire parameter vector is used, not just with respect to the parameter of interest. This takes account (to some degree) of expected changes in other, correlated, parameters due to freeing the constraint  \citep{sorbom1989model}. The EPC's can be used to obtain an estimate of the difference between groups from multiple group SEM as the difference of the EPC's $\delta_g$ for each group \citep{bentler1984efficient}:
$(\hat{\theta}+\delta_2) - (\hat{\theta}+\delta_1) = \delta_2 - \delta_1$.

\newtheorem{theorem}{Theorem}
\begin{theorem}
Formulate an MG-SEM with equality restrictions on all parameters.
Order groups by size of $\theta_g$. Define $\hat{\IPC}$ as in (\ref{eq:solution}). Then:
$\delta_2 - \delta_1 = \bar{\hat{\IPC}}_2 - \bar{\hat{\IPC}}_1$.
\end{theorem}

\begin{proof}
The estimated difference between groups from MG-SEM is: 
$J_{2}^{-1} (\partial F / \partial \theta_2) - J_{1}^{-1} (\partial F / \partial \theta_1) = 
J^{-1} [(\partial F / \partial \theta_2) - (\partial F / \partial \theta_1)]$
because $J_{1} = J_{2}$ since all parameters estimates are equal and $\theta_1$ and $\theta_2$ play the same role in different groups:  $\Delta_{\theta_1} = \Delta_{\theta_1}$ and $V_1 = V_2$ so $J_1 = J_2 = J$.

Now, $\partial F / \partial \theta_g =  \Delta' V [s_g - \sigma(\theta_g)]$ \citep[e.g.][]{neudecker1991linear}. So 
$(\partial F / \partial \theta_2) - (\partial F / \partial \theta_1) = \Delta' V [s_2 - s_1]$, again because $\sigma(\hat{\theta}_2) = \sigma(\hat{\theta}_1)$. Since by definition $s_g = \bar{d}_g$, $\Delta' V [s_2 - s_1] = \Delta' V [\bar{d}_2 - \bar{d}_1]$, leading to 
$\delta_2 - \delta_1 = J^{-1} \Delta' V [\bar{d}_2 - \bar{d}_1] = 
%n_2^{-1} \sum_{G_2}{(W \cdot d)} - n_1^{-1} \sum_{G_1}{(W \cdot d)} = 
\bar{\hat{\IPC}}_2 - \bar{\hat{\IPC}}_1$.
\end{proof}

The ``modification index'' (score test) for freeing this restriction is defined as 
the difference divided by the variance of the difference. Some standard software packages currently use the inverse hessian $J^{-1}$ as the variance. However, from \citet[section 5]{satorra1989alternative} it is clear that this is only valid when asymptotically optimal (AO) estimation has been used. For non-AO estimation, such as pseudo-maximum likelihood or in the case of nonnormally distributed $\y$, a sandwich estimator is needed \citep{satorra1989alternative}. The resulting hypothesis test is then equivalent to the so-called ``generalized'' score test \citep{boos1992generalized}. From $\bar{\hat{\IPC}}_2 - \bar{\hat{\IPC}}_1 = J^{-1} \Delta' V [\bar{d}_2 - \bar{d}_1]$, it is clear that a standard $z$-test performed on the difference in IPC's fulfills this requirement since it will automatically calculate the sandwich variance $\var(\delta_2 - \delta_1) = J^{-1} \Delta' V \var(\bar{d}_2 - \bar{d}_1) V \Delta J^{-1}$. Thus, squared $z$-statistics from the regression of $\hat{\IPC}$ on a group indicator can be interpreted as generalized score tests that are robust to nonnormality.

Complex sampling can be taken into account as described in the text: by adjusting $d$ using the weights and estimating $\var(\bar{d}_2 - \bar{d}_1)$ incorporating clustering and stratification identifiers \citep{muthen_complex_1995}.
In general MG-SEM can be expected to yield the most accurate results. However, since some SEM software does not correct the modification indices for nonnormality or complex sampling, in cases of strong nonnormality or complex sampling design effects, the proposed procedure may therefore actually be preferable to MG-SEM.

\section{R code providing the individual parameter contributions}\label{sec:Rcode}
%\singlespacing
\small
\begin{verbatim}
# Obtain the derivative of SEM parameters w.r.t. sigma
get.g <- function(fit) {
  # d sigma / d theta
  D.free <- lavaan:::computeDelta(fit@Model)[[1]]
  colnames(D.free) <- names(coef(fit))

  # NT weight matrix
  V.nt <- lavaan:::getWLS.V(fit)[[1]]

  H <- solve(t(D.free) %*% V.nt %*% D.free)
  g <- H %*% t(D.free) %*% V.nt

  g
}

# Obtain the rescaled individual moment contributions.
get.di <- function(x, center=NULL) {
  if(is.null(center)) center <- colMeans(x, na.rm=TRUE)
  xd <- t(t(x) - center)
  
  p <- NCOL(xd); n <- nrow(xd)
  idx.vech <- vech(matrix(1:(p^2), p)) # Remove redundant elts.
  a <- matrix(rep(xd, p), ncol = p * p)[,idx.vech]
  b <- xd[, rep(1:p, each = p)][,idx.vech]
  d <- a * b * n/(n - 1)

  cbind(x, d) # add means
}

# Obtain the IPC dataset given 
#     a matrix of observations X and a lavaan fit object.
get.ipc <- function(X, fit) {
  g <- get.g(fit)
  d <- t(get.di(X))
  wi <- as.data.frame(-t(g %*% d)) # assumes meanstructure=TRUE.

  wi
}
\end{verbatim}

\section{Example application}\label{sec:Rcode-example}
\begin{verbatim}
library(lavaan)
library(systemfit) # Convenient for regression with many dep. variables
# The Holzinger and Swineford (1939) example, single-group model
HS.model <- " visual  =~ x1 + x2 + x3
              textual =~ x4 + x5 + x6
              speed   =~ x7 + x8 + x9 "
# Fit the model
fit <- lavaan(HS.model, data=HolzingerSwineford1939,
              auto.var=TRUE, auto.fix.first=TRUE, 
              auto.cov.lv.x=TRUE, 
              meanstructure=TRUE, int.ov.free=TRUE) # necessary 
# Obtain the IPC's and join them with the original data
ipc.data <- cbind(HolzingerSwineford1939, get.ipc(fit@Data@X[[1]], fit))

# Use systemfit for convenience to regress parameters on covariates.
equation.list <- lapply(names(coef(fit)), function(parname) {
  as.formula(sprintf("`%s` ~ sex+ageyr+agemo+school+grade", parname))
}) # Generates list of formulas
names(equation.list) <- names(coef(fit))

fit.system.ols <- systemfit(equation.list, data=ipc.data)
summary(fit.system.ols)

# Grade appears to be an important grouping variable 
\end{verbatim}
\clearpage
\section{Internet use study results}

Extended results for the application of IPC regression to the internet use example are given in this section.

\begin{table}[tb]
\centering
\caption{Coefficient estimates for regression of the IPC of the error variance parameter on respondent characteristics ($n = 2838$). The $R^2$ is 0.03.}
\label{tab:application-regression}
\begin{tabular}{llrrrr}
  \hline
 && Estimate & s.e. & $t$-value & $p$-value \\ 
  \hline
 & (Intercept) 			& 0.08 & (0.09) & 0.93 & 0.35 \\ 
 & Self-employed 			& 0.09 & (0.06) & 1.40 & 0.16 \\ 
 & Year birth 			& 6.71 & (0.91) & 7.34 & 0.00 \\ 
 & Year birth$^2$ 		& -2.33 & (0.91) & -2.55 & 0.01 \\ 
 & Female 				& -0.08 & (0.03) & -2.42 & 0.02 \\ 
 \multicolumn{2}{l}{\emph{Education}}\\
 & Primary school		& 0.00\\
&VMBO 			& 0.17 & (0.09) & 1.90 & 0.06 \\ 
 &MBO 			& 0.26 & (0.09) & 2.82 & 0.00 \\ 
  & HAVO/VWO 		& 0.18 & (0.10) & 1.83 & 0.07 \\ 
  & HBO 			& 0.29 & (0.09) & 3.21 & 0.00 \\ 
  & University 	& 0.37 & (0.11) & 3.54 & 0.00 \\ 
   \hline
\end{tabular}
\end{table}

\begin{table}[tb]
\centering
\caption{Full mediation SEM: error variance IPC (dependent variable), ``all-zeroes'' indicator (mediator), and covariates (predictors). Satorra-Bentler chi-square is 11 with 9 degrees of freedom ($p = 0.276$, scaling correction = 0.953); CFI = 0.999, RMSEA = 0.009, $n = 2838$.}
\label{tab:application-mediation}
\begin{tabular}{llrrr}
  \hline
	  \multicolumn{2}{l}{Parameter}  & Est. & s.e. & $z$ \\ 
  \hline
\multicolumn{2}{l}{\emph{Regression coefficients}}\\
\multicolumn{2}{l}{Error variance IPC on} \\
	& All-zeroes & -0.50 & 0.03 & -17.0 \\ \\
\multicolumn{2}{l}{All-zeroes on}\\
	&  Self-employed & -0.06 & 0.03 & -2.4 \\ 
	    & Year birth & -10.07 & 0.40 & -25.1 \\ 
    & Year birth$^2$ & 6.46 & 0.41 & 16.0 \\ 
	    & Female & 0.10 & 0.02 & 6.2 \\ 
    & Education: VMBO & -0.13 & 0.04 & -3.1 \\ 
	    & Education: MBO & -0.31 & 0.04 & -7.4 \\ 
	    & Education: HAVO/VWO & -0.35 & 0.04 & -7.8 \\ 
	    & Education: HBO & -0.46 & 0.04 & -11.3 \\ 
	    & Education: University & -0.51 & 0.04 & -11.4 \\ \\ 
\multicolumn{2}{l}{\emph{Variance parameters}}\\
&  Error variance IPC & 0.75 & 0.06 & 11.8 \\ 
&  All-zeroes & 0.17 & 0.00 & 50.1 \\ 

   \hline
\end{tabular}
\end{table}

\end{document}